\documentclass[11pt,a4paper,reqno,oneside]{amsart}
\usepackage[utf8]{inputenc}
\usepackage[T1]{fontenc}
\usepackage[a4paper]{geometry}
\usepackage{lmodern,textcomp,amsmath,amssymb,mathtools,braket,authblk,amsthm,csquotes,bm}
\usepackage[english]{babel}
\usepackage[foot]{amsaddr}

\usepackage{filecontents} %for arXiv submission

\newtheorem{theorem}{Theorem}[section]

\newtheorem{lemma}[theorem]{Lemma}
\newtheorem{proposition}[theorem]{Proposition}
\theoremstyle{definition}
\newtheorem{assumption}[theorem]{Assumption}
\newtheorem{definition}[theorem]{Definition}
\theoremstyle{remark}
\newtheorem{remark}[theorem]{Remark}
\newtheorem{example}[theorem]{Example}
\AtBeginEnvironment{example}{%
	\pushQED{\qed}%
}
\AtEndEnvironment{example}{\popQED\endexample}
\AtBeginEnvironment{remark}{%
	\pushQED{\qed}%
}
\AtEndEnvironment{remark}{\popQED\endexample}

\newcommand{\hilb}{\mathcal{H}}
\newcommand{\fock}{\mathcal{F}}
\newcommand{\hfrak}{\mathfrak{H}}
\newcommand{\e}{\mathrm{e}}
\newcommand{\g}{\mathrm{g}}
\renewcommand{\Im}{\operatorname{Im}}

\renewcommand{\a}[1]{a\!\left(#1\right)}
\newcommand{\adag}[1]{a^*\!\left(#1\right)}

\newcommand{\dOmega}{\mathrm{d}\Gamma(\omega)}
\newcommand{\dom}{\operatorname{Dom}}
\newcommand{\ran}{\operatorname{Ran}}
\renewcommand{\ker}{\operatorname{Ker}}
\renewcommand{\span}{\operatorname{Span}}

\makeatletter
\def\smalloverbrace#1{\mathop{\vbox{\m@th\ialign{##\crcr\noalign{\kern3\p@}%
				\tiny\downbracefill\crcr\noalign{\kern3\p@\nointerlineskip}%
				$\hfil\displaystyle{#1}\hfil$\crcr}}}\limits}
\makeatother

\usepackage[unicode=true,
bookmarks=true,bookmarksopen=false,
breaklinks=false,pdfborder={0 0 0},colorlinks=true]
{hyperref}

\usepackage{xcolor}
\definecolor{cblue}{rgb}{0.16, 0.32, 0.75}
\definecolor{cred}{rgb}{0.7, 0.11, 0.11}
\hypersetup{%
	,linkcolor=cred
	,citecolor=cblue
	,urlcolor=black
}

\usepackage[
backend=bibtex,
style=numeric-comp,
giveninits=true,
maxbibnames=299,
natbib=true,
doi=false,
isbn=false,
url=false,
sorting=none
]
{biblatex}
\renewbibmacro{in:}{}
\DeclareFieldFormat[misc]{title}{``#1''\isdot}
\DeclareFieldFormat{journaltitle}{#1\isdot}
\DeclareFieldFormat[article,periodical]{volume}{\mkbibbold{#1}}
\DeclareFieldFormat{pages}{#1}
\AtEveryBibitem{%
	\clearfield{number}}

\usepackage{orcidlink}

\newcommand\blfootnote[1]{%
  \begingroup
  \renewcommand\thefootnote{}\footnote{#1}%
  \addtocounter{footnote}{-1}%
  \endgroup
}

\title[Generalized spin--boson models with mild ultraviolet divergences]{Self-adjointness and domain of generalized spin--boson models with mild ultraviolet divergences}

\author{Sascha Lill\textsuperscript{\,1}\hspace{3pt}\orcidlink{0000-0002-9474-9914}\hspace{1pt}}
\author{Davide Lonigro\textsuperscript{\,2}\hspace{3pt}\orcidlink{0000-0002-0792-8122}\hspace{1pt}}

\address{\footnotesize \textsuperscript{1}Dipartimento di Matematica Federigo Enriques, Universit\`a degli Studi di Milano, I-20133 Milan, Italy}
\address{\footnotesize \textsuperscript{2}Department Physik, Friedrich-Alexander-Universität Erlangen-Nürnberg, Staudtstraße 7, 91058 Erlangen, Germany}
\email{\footnotesize \href{mailto:sascha.lill@unimi.it}{\texttt{sascha.lill@unimi.it}}}
\email{\footnotesize \href{mailto:davide.lonigro@fau.de}{\texttt{davide.lonigro@fau.de}}}

\begin{document}
	
	\maketitle
 \thispagestyle{empty}

 \vspace{-0.5cm}
	\begin{abstract}
		We provide a rigorous construction of a large class of generalized spin--boson models with ultraviolet-divergent form factors. This class comprises various models of many possibly non-identical atoms with arbitrary but finite numbers of levels, interacting with a boson field. Ultraviolet divergences are assumed to be mild, such that no self-energy renormalization is necessary. Our construction is based on recent results by A. Posilicano, which also allow us to state an explicit formula for the domain of self-adjointness for our Hamiltonians.
	\end{abstract}

    \blfootnote{2020 \textit{Mathematics Subject Classification}. 46N50, 47A10, 47B25, 81Q10, 81T16.}

    \noindent \small \textbf{Keywords}: Spin--boson model; renormalization; ultraviolet divergence; self-adjointness domain; interior--boundary conditions; scales of Hilbert spaces. \normalsize
    
	\section{Introduction and main result}\label{sec:intro}
	
	We consider a quantum mechanical system, described by a common finite--dimensional Hilbert space $\mathfrak{h} := \mathbb{C}^D$, coupled to a single boson field, described by an abstract single-particle Hilbert space $\hilb$. The Hilbert space of the entire system is given by
    \begin{equation}
        \hfrak := \mathfrak{h} \otimes \fock, \qquad
        \fock := \fock^+(\hilb) = \bigoplus_{n = 0}^\infty S_n \hilb^{\otimes n},
    \end{equation}
    with $S_n$ being the symmetrization operator, and $\fock$ the \textit{bosonic Fock space}. The Hamiltonian of the non-interacting system is of the form
    \begin{equation}\label{eq:hfree}
        H_{\textnormal{free}} := K\otimes I_{\fock}+I_{\mathfrak{h}}\otimes\dOmega,
    \end{equation}
    where $K\in\mathbb{C}^{D \times D}$ is a self-adjoint matrix describing the energy of the atoms, and $\dOmega$ is the second quantization of some positive self-adjoint operator $\omega\geq m_0>0$ defined on a dense subspace of $\hilb$. $m_0$ is also called the \textit{mass of the boson}. Without loss of generality, and for physical convenience, we will fix $\hilb$ as the space of square-integrable complex-valued functions on some measure space $(X,\Sigma,\mu)$, and $\omega$ as the multiplication by some measurable real-valued function $k\in X\mapsto\omega(k)\in\mathbb{R}$.

    The interaction between the quantum system and the field will be constructed as follows. Given $N\in\mathbb{N}$, we take a family of matrices $B_j \in \mathbb{C}^{D \times D}, j \in \{1, \ldots, N\}$, and a family of form factors $f_j, j \in \{1, \ldots, N\}$, each being a measurable function $f_j : X \to \mathbb{C}$. We consider an interaction term in the form $A^* + A$, where, formally,
    \begin{equation} \label{eq:AAdag}
        A := \sum_{j = 1}^N B_j^* \otimes a(f_j), \qquad
        A^* := \sum_{j = 1}^N B_j \otimes a^*(f_j),
    \end{equation}
    with $a^*(f_j), a(f_j)$ (formally) being the bosonic creation/annihilation operators on $\fock$ defined by requiring
    \begin{equation} \label{eq:adaggera}
    \begin{aligned}
        a^*(f_j) S_n (\phi_1 \otimes \ldots \otimes \phi_n) := &\sqrt{n+1} \,S_n(f_j \otimes \phi_1 \otimes \ldots \otimes \phi_n),\\ 
        a(f_j) S_n (\phi_1 \otimes \ldots \otimes \phi_n) := &\frac{1}{\sqrt{n}} S_n\sum_{\ell = 1}^n \langle f_j, \phi_\ell \rangle (\phi_1 \otimes \ldots \otimes \phi_{\ell-1} \otimes \phi_{\ell+1} \otimes \ldots \otimes \phi_n)
    \end{aligned}
    \end{equation}
   for any $ \phi_1, \ldots, \phi_n \in \hilb $ and $n\ge0$ (for $a^*(f_j)$) or $n\ge1$ (for $a(f_j)$), respectively. One easily checks that these operators satisfy the canonical commutation relations (CCR)
   \begin{equation} \label{eq:CCR}
        [a(f), a(g)] = [a^*(f), a^*(g)] = 0, \qquad
        [a(f), a^*(g)] = \langle f, g \rangle.
   \end{equation}
   Thus, the total Hamiltonian formally reads
    \begin{equation}\label{eq:model}
    \begin{aligned}
		H = H_{f_1,\dots,f_N}
        :=&H_{\textnormal{free}} + A^* + A\\
        =&H_{\textnormal{free}} + \sum_{j=1}^N\left[B_j\otimes\adag{f_j} + B_j^*\otimes\a{f_j}\right].
    \end{aligned}
	\end{equation}
    Such models, denoted as \textit{generalized spin--boson} (GSB) \textit{models}, were first introduced, to our best knowledge, in~\cite{arai1997existence} as a generalization of the well-known spin--boson model; their mathematical properties have been extensively addressed in the literature, cf.~\cite{hirokawa2001remarks,hubner1995spectral,hirokawa1999expression,arai1990asymptotic,amann1991ground,davies1981symmetry,fannes1988equilibrium,hubner1995radiative,reker2020existence,hasler2021existence} for the single spin case and~\cite{arai2000essential,arai2000ground,falconi2015self,takaesu2010generalized,teranishi2015self,teranishi2018absence} for the general one.
 
    The aim of this work is to give, for a general choice of $(N,D)$ and of the operators $B_1,\dots,B_N$, a rigorous interpretation of Eq.~\eqref{eq:model} as a self-adjoint operator on a dense subspace of $\hfrak$ in the case in which the form factors $f_1,\dots,f_N$ are not square-integrable---a problem that is also called \textit{non-perturbative renormalization problem}. At this point, we did not specify yet 
    the integrability properties of the $f_j$'s.
    And indeed, the difficulty of the renormalization problem heavily depends on the growth properties of $f_j$.
    
    To elucidate how the difficulty of the renormalization problem depends on the growth of $f_j$, let us state some well-known results in the case $N=D=1$, also known as the \textit{van Hove model}~\cite{derezinski2003vanHove}. In that case, we only have one  form factor $f : X \to \mathbb{C}$. Depending on the behavior of $f$ with respect to the dispersion relation $\omega$, one can now essentially distinguish 4 difficulty classes for the renormalization problem.\medskip
    
    \paragraph*{\textbf{Case 0}:} $\int|f|^2 < \infty \Leftrightarrow f \in \hilb$. This is the simplest case. Both $A^*$ and $A$ are Kato perturbations of $H_{\textnormal{free}}$, so Eq.~\eqref{eq:model} readily defines a self-adjoint operator $H$ on $\dom (H) = \dom(H_{\textnormal{free}})$.\medskip
    
    \paragraph*{\textbf{Case 1}:} $\int|f|^2 = \infty$, but still, $\int\frac{|f|^2}{\omega} < \infty$. In this case, $A^*+A$ is still a form perturbation of $H_{\textnormal{free}}$, so the KLMN theorem allows defining $H$ as a self-adjoint operator. However, the perturbation may change the operator domain.\medskip
    
    \paragraph*{\textbf{Case 2}:} $\int\frac{|f|^2}{\omega} = \infty$, but still, $\int\frac{|f|^2}{\omega^2} < \infty$. Here, we can no longer make sense of Eq.~\eqref{eq:model} without further modifications. In fact, one needs a \textit{self-energy renormalization procedure}, to obtain a well-defined Hamiltonian: for some sequence of UV-cutoffs $\Lambda \in \mathbb{R}$, one introduces the regularized form factors $f_\Lambda \in \hilb$ such that $\lim_{\Lambda \to \infty} f_\Lambda = f$ in a suitable sense. Substituting $f$ by $f_\Lambda$ in $H$ then yields the regularized Hamiltonians $H_\Lambda$, which are well-defined on a dense subspace of $\fock$. By subtracting from them the \textit{self-energies} $ E_\Lambda := -\int\frac{|f_\Lambda|^2}{\omega} \in \mathbb{R} $, one may then obtain a renormalized Hamiltonian:
    \begin{equation}
        \widetilde{H} := \lim_{\Lambda \to \infty} (H_\Lambda - E_\Lambda),
    \end{equation}
    which is self-adjoint on a dense domain in $\fock \simeq \hfrak$.
    
    In fact, using a (unitary) \textit{dressing operator} $W_\Lambda := \exp \left( \a{\omega^{-1}f_\Lambda} - \adag{\omega^{-1}f_\Lambda} \right)$, it turns out that
    \begin{equation}
        W_\Lambda^* H_\Lambda W_\Lambda = H_{\textnormal{free}} + E_\Lambda,
    \end{equation}
    so $W_\Lambda$ extracts the self-energy. In particular,
    \begin{equation} \label{eq:directdressing}
        \widetilde{H} = W H_{\textnormal{free}} W^*,
    \end{equation}
    with the unitary dressing operator $W := \exp \left( \a{\omega^{-1}f} -\adag{\omega^{-1}f} \right)$, so the dynamics generated by $\widetilde{H}$ are unitarily equivalent to those of $H_{\textnormal{free}}$.\medskip
    
    \paragraph*{\textbf{Case 3}:} $\int\frac{|f|^2}{\omega^2} = \infty$. This is the most singular case. Formally, Eq.~\eqref{eq:directdressing} would yield a renormalized operator $\widetilde{H}$ that is equivalent to $H_{\textnormal{free}}$, but since $W$ is ill-defined, Eq.~\eqref{eq:directdressing} makes mathematically no sense.
    One way out of this problem is to interpret $H$ as in Eq.~\eqref{eq:model} as an element of a suitably defined $^*$-algebra $\mathcal{A}$ and to replace the conjugation with $W$ by an (algebraic) \textit{Weyl transformation} $\mathcal{V}_W : \mathcal{A} \to \mathcal{A}$ that maps $H \mapsto H_{\textnormal{free}}$. As $\mathcal{V}_W$ preserves the commutation relations, one may see $H_{\textnormal{free}}$ as a legitimate renormalized Hamiltonian, which is self-adjoint on a dense domain in $\fock$~\cite[Sect.~3]{FewsterRejzner2020introAQFT}.\footnote{More precisely, one considers a subalgebra $\mathcal{A}_{\rm obs} \subset \mathcal{A}$, containing all relevant observables, and constructs a representation $\pi$ of $\mathcal{A}_{\rm obs}$ on some domain $\mathcal{D} \subset \fock$, in which $\pi(H)=H_{\textnormal{free}}$. So, $H$ acts as if it were a free Hamiltonian. However, in this representation, $\pi(a^*(\varphi)), \pi(a(\varphi))$ for $ \varphi \in \hilb$ (if they are defined) do not act like creation/annihilation operators in Eq.~\eqref{eq:adaggera}, and there is also no unitary operator $W : \fock \to \fock$ that makes $W \pi(a^*(\varphi)) W^*, W \pi(a(\varphi)) W^*$ act like creation/annihilation operators. So $\pi$ is inequivalent to the standard Fock representation. Equivalently, one also says that $\pi$ is a non-Fock representation, or that $H$ leaves the Fock space.}
    
    Another way out is to interpret $W$ as a map from a subspace of $\fock$ to some Fock space extension $\overline{\fock}$, and to define $E := -\int\frac{|f|^2}{\omega}$ as an element of a suitable vector space~\cite{lill2022time, lill2020extended}. This way, one can give a precise mathematical meaning to the equation
        \begin{equation}
            H_{\textnormal{free}} = W (H - E) W^*,
        \end{equation}
    and interpret $H_{\textnormal{free}}$ as the renormalized Hamiltonian.\medskip

    Coming back to the general case, let us mention that in \textbf{Case 0}, i.e., $f_1,\dots,f_N\in\hilb$ for $N, D \in \mathbb{N}$, the renormalization problem is easily solved: $A^* + A$ is again a Kato perturbation of $H_{\textnormal{free}}$, so $H$ as in Eq.~\eqref{eq:model} is self-adjoint with coupling-independent domain $\dom (H)=\dom (H_{\textnormal{free}})=\mathfrak{h}\otimes\dom(\dOmega)$.

    In this work, we address the renormalization problem for general values of $N, D \in \mathbb{N}$, and form factors $f_1,\dots,f_N$ corresponding to \textbf{Case 1}. Our proofs essentially rely on a recent work by Posilicano~\cite{posilicano2020self}, which introduces an abstract machinery for renormalizing Hamiltonians of the form $H_{\textnormal{free}}+A^*+A$ and finding their domains. This paper thus demonstrates how these results can be successfully applied to generalized spin--boson models. To our best knowledge, the present work is the first rigorous proof of a renormalization in Case 1 for this kind of models with a general choice of $B_1,\dots,B_N$: while the existence of a self-adjoint operator associated with Eq.~\eqref{eq:model} is in fact guaranteed by the KLMN theorem (cf.~\cite[Proposition~4.2]{lonigro2022generalized}), explicit expressions for the domain were only obtained in particular cases~\cite{lonigro2022generalized,lonigro2022renormalization,lonigro2023self}.

    Within \textbf{Case 2}, the application of the abstract results from~\cite{posilicano2020self}, while still possible, offers new challenges---as it will be elucidated in Section~\ref{sec:posilicano} for the abstract case, an unambiguous identification of a self-adjoint operator associated with Eq.~\eqref{eq:model} is no longer viable. One needs to specify a suitably ``regularizing operator'' $T$ which must be both mathematically admissible and physically meaningful. This problem will be investigated elsewhere. It would be highly desirable to also achieve self-adjointness results in \textbf{Case 3}; however, the construction of a suitable dressing transformation $W$ for general values of $N$ and $D$ is far more challenging than within the van Hove model. We hope to achieve this goal in future research.
    
    For other models of matter--field interaction, there exists an extensive mathematical literature on the non-perturbative renormalization problem. Some references about successful renormalization schemes in Cases 1 and 2 include~\cite{nelson1964interaction, eckmann1970model, frohlich1973infrared, sloan1974polaron, griesemer2016self, griesemer2018domain}. A successful renormalization in Case 3 was achieved for a related model by Gross~\cite{gross1973relativistic}. We refer to~\cite[Sect.~1.3]{lill2022time} for a more detailed discussion of the related literature. We also point out that there exists a recently introduced renormalization technique based on the so-called ``interior--boundary conditions'' (IBC)~\cite{teufel2021hamiltonians, teufel2016avoiding, lampart2018particle, lampart2019nelson, lampart2019nonrelativistic, schmidt2021massless, lampart2020renormalized}, which is closely related to the work of Posilicano and ours. IBC has successfully rendered renormalized Hamiltonians together with their domains in Cases 1 and 2. For a more detailed overview about the recent IBC literature, we refer the reader to~\cite[Sect.~1.4]{lill2022time}.	

    Furthermore, the renormalization problem has been addressed and solved in the context of Constructive Quantum Field Theory (CQFT), e.g.,~\cite{glimm1968lambda, glimm1970lambda, glimm1970lambdaIII, glimm1971yukawa2, glimm1973positivity}. Here, Haag's Theorem~\cite{haag1955quantum, haag1996local} requires to ``leave the Fock space'', as in Case 3 above, while relativistic causality admits the employment of further techniques (e.g., Segal's theorem~\cite{segal1967notes}, which circumvents the search for a dressing transformation). We refer to~\cite{dedushenko2023snowmass, summers2012perspective} for an extensive overview of the literature on CQFT.
    
    For generalized spin--boson models, there is a plenitude of applications in the physics literature, often requiring UV-divergent form factors~\cite{leggett1987dynamics,breuer2002theory,weiss2012quantum}, common examples being the description of subradiance and superradiance phenomena in quantum optics~\cite{dicke1954coherence,gross1982superradiance,van2013photon}, e.g., bound states in the continuum in photonic waveguides~\cite{dorner2002laser,tufarelli2013dynamics,sanchez2017dynamical,gonzalez2017efficient,facchi2019bound,lonigro2021stationary}, or toy models of quantum Markovian systems~\cite{lonigro2022quantum,lonigro2022regression}.  On the mathematical side, there is a rich literature about the spectrum and dynamics of the spin--boson model, see e.g.,~\cite{fannes1988equilibrium, arai1999absence, derezinski2001spectral, jakvsic2006mathematical, bach2017existence, teranishi2018absence, hasler2021existence} and the references therein; however, the form factors often contain a UV-cutoff and the renormalization problem is not addressed. Concerning the renormalization of the spin--boson model, we point out a recent result by Dam and M\o ller~\cite{dam2020asymptotics}, proving that certain spin-boson models with $\mathfrak{h} = \mathbb{C}^2$ in Case 3 renormalize to a direct sum of two van Hove Hamiltonians for an appropriately chosen dressing transformation. Finally, an explicit renormalization of a particular class of spin--boson models was recently addressed by one of the authors, first for $N=1$ in Cases~1 and~2~\cite{lonigro2022generalized, lonigro2022renormalization} and then for $N>1$ in Case~1~\cite{lonigro2023self}. However, the present article is, up to our best knowledge, the first work establishing self-adjointness \textit{and} a formula for the domain in the general case.

    To achieve this result, we will resort to the formalism of \textit{scales of Hilbert spaces}~\cite{albeverio2000singular,albeverio2007singularly,simon1995spectral}. Essentially, one introduces two scales of spaces $\{\hilb_s\}_{s\in\mathbb{R}},\{\fock_s\}_{s\in\mathbb{R}}$, obtained as the completion of a space of sufficiently well-behaved vectors under the norms
	\begin{equation} \label{eq:hilbscale}
		\left\|f\right\|_{\hilb_s}:=\left\|\omega^{s/2}f\right\|_{\hilb},\qquad\left\|\Psi\right\|_{\fock_s}:=\left\|\left(\dOmega+1\right)^{s/2}\Psi\right\|_{\fock}.
	\end{equation}
	Of course, for $s=0$ one obtains the original spaces, and $\hilb_{+s}\subset\hilb\subset\hilb_{-s}$, as well as $\fock_{+s}\subset\fock\subset\fock_{-s}$ for any $s\geq0$, with all inclusions being dense. Mathematically, $\hilb_{-s}$ is a space of non-normalizable functions which, however, satisfy a growth constraint weaker than the one of $L^2$: a more negative $-s$ corresponds to stronger allowed growths where $\omega(k)$ is large, that is, to stronger UV-singularities. In fact, we can conveniently express the divergence case in terms of the index $s$:
    \begin{itemize}
        \item \textbf{Case 0}: $f_1,\dots,f_N \in \hilb = \hilb_0$;
        \item \textbf{Case 1}: $f_1,\dots,f_N \in \hilb_{-1} \setminus \hilb_0$;
        \item \textbf{Case 2}: $f_1,\dots,f_N \in \hilb_{-2} \setminus \hilb_{-1}$;
        \item \textbf{Case 3}: $f_1,\dots,f_N \notin \hilb_{-2}$.
    \end{itemize}	
	Within this formalism, given $s\geq1$ and $f\in\hilb_{-s}$, one can define two continuous linear operators $\a{f}\in\mathcal{B}(\fock_{+s},\fock)$ and $\adag{f}\in\mathcal{B}(\fock,\fock_{-s})$, whose action is compatible with the ``regular'' creation and annihilation operators~\eqref{eq:adaggera} for $f\in\hilb$, and which can be approximated, in the respective operator norms, by sequences of regular creation and annihilation operators (see e.g.,~\cite[Proposition~3.7]{lonigro2022generalized}). 
 
	\subsection{Main result}

    Our results are based on a recent work by Posilicano~\cite{posilicano2020self}, which provides abstract conditions for the self-adjointness of Hamiltonians corresponding to the formal expression $H_{\textnormal{free}}+A^*+A$, see Section~\ref{sec:posilicano}. The conditions on the matrices $B_j$ under which the results of~\cite{posilicano2020self} apply to our Hamiltonian $H$ in Eq.~\eqref{eq:model} can be concluded as follows:
    \begin{assumption} \label{as:Bj}
    The interaction matrices $B_j$, $j \in \{1, \ldots, N\}$, are chosen such that:
    \begin{itemize}
        \item[(i)] every $B_j$ is normal, i.e., $B_j^* B_j=B_jB_j^*$,
        \item[(ii)] $[B_j,B_\ell]=0$, $j,\ell=1,\dots,N$, and
        \item[(iii)] $\bigcap_{j=1}^N\ker B_j=\{0\}$.
    \end{itemize}
    \end{assumption}
    
    As discussed, we will examine the case $f_1,\dots,f_N \notin \hilb$. In fact, we will even require the stronger assumption that no linear combination of the $f_j$ lies in $\hilb$ (except for the trivial one, of course).
    \begin{definition}[{$\hilb$-independence;~\cite[Def.~3.1.1]{albeverio2000singular}}] \label{def:hilbindepentent}
        A family $f_1, \ldots, f_N \in \hilb_{-s}\setminus\hilb, s > 0$ is called $\hilb$-\textit{independent}, if and only if for all $c_1, \ldots, c_N \in \mathbb{C}$,
        \begin{equation}\label{eq:hilbindependence}
            \sum_{j=1}^N c_j f_j \notin \hilb, \quad \text{unless } c_j = 0 \; \forall j.
        \end{equation}
    \end{definition}
    %It is easy to see that any $\hilb$-independent family is also a family of linearly independent elements of $\hilb_{-s}$.
    %\textit{A fortiori}, an $\hilb$-independent family of form factors is also a family of linearly independent elements of $\hilb_s$.
    
    Also notice that there is no loss of generality in requiring the operators $B_j$'s to be linearly independent: if this is not the case, i.e., if there exist $\alpha_1,\dots,\alpha_{N-1}\in\mathbb{C}$ such that (up to permutations) $B_N=\sum_{j=1}^{N-1}\alpha_jB_j$, then
    \begin{eqnarray}\label{eq:rearranging}
        A=\sum_{j=1}^{N-1}B_j^*\otimes\a{f_j}+\sum_{j=1}^{N-1}\overline{\alpha_j}B_j^*\otimes\a{f_N}\nonumber
        =\sum_{j=1}^{N-1}B_j^*\otimes\a{f_j+\alpha_jf_N}.
    \end{eqnarray}
    This request is compatible with the $\hilb$-independence of the form factors: given an $\hilb$-independent set $\{f_j\}_{j=1,\dots,N}$, the set of form factors after the rearrangement, $\{f_j+\alpha_jf_N\}_{j=1,\dots,N-1}$, is still $\hilb$-independent. 
      
    \begin{theorem}[Main Result]\label{thm:main}
	Let $f_1, \ldots, f_N \in \hilb_{-s} \setminus \hilb$, $s\le1$, be a family of $\hilb$-independent form factors and let $B_1, \ldots, B_N \in \mathbb{C}^{D \times D}$ satisfy Assumption~\ref{as:Bj}. Then the following statements hold true:
    \begin{itemize}
        \item[(a)] The operator $H = H_{\textnormal{free}} + A^* + A$ as in Eq.~\eqref{eq:model}, with $H_{\textnormal{free}}$ as in Eq.~\eqref{eq:hfree} and $A,A^*$ as in Eq.~\eqref{eq:AAdag}, is well-defined and self-adjoint on the domain
        \begin{equation} \label{eq:Hdomain}
            \dom (H) = \left\{ \Psi \in \hfrak_{2-s} \; : \; \left[ 1 + (H_{\textnormal{free}}-z_0)^{-1}A^* \right] \Psi \in \dom (H_{\textnormal{free}}) \right\},
        \end{equation}
        where $z_0\in\rho(H_{\textnormal{free}})\cap\mathbb{R}$.
        \item[(b)] For every $j \in \{1, \ldots, N\}$, there exists an approximating sequence $\{f_{j, n}\}_{n \in \mathbb{N}} \subset \hilb$ with $\Vert f_{j, n} - f_j \Vert_{\hilb_{-s}} \to 0$ as $n \to \infty$, such that, defining
        \begin{equation} \label{eq:AnAndagger}
            A_n := \sum_{j=1}^N B_j^* \otimes a(f_{j, n}), \qquad
            A_n^* := \sum_{j=1}^N B_j \otimes a^*(f_{j, n}),
        \end{equation}
        we have
        \begin{equation}\label{eq:normresconvergence}
            (H_{\textnormal{free}} + A_n^* + A_n) \to (H_{\textnormal{free}} + A^* + A)
        \end{equation}
        as $n\to\infty$ in the norm resolvent sense.
    \end{itemize}
	\end{theorem}
    In the above statement, it is understood that the annihilation and creation operators in the definition of $A$ and $A^*$, whenever involving form factors $f_j\notin\hilb$, are to be interpreted as elements, respectively, of $\mathcal{B}(\fock_{+1},\fock)$ and $\mathcal{B}(\fock,\fock_{-1})$, cf.~Proposition~\ref{prop:af}.

    \begin{remark}[Identical atom spaces] \label{rem:spinspace}
        In the case in which the quantum system is composed by $N$ $d$-dimensional smaller systems (atoms) each separately interacting with the boson field, it is common to associate to each atom a Hilbert space $\mathbb{C}^d$, so $\mathfrak{h}=\mathbb{C}^{d^N}$, i.e., $D=d^N$. Typically, $B_j$ then only acts on the tensor factor of the $j$-th spin, which means it takes the form
        \begin{equation}\label{eq:identical}
            B_j = I_{\mathbb{C}^d} \otimes \ldots \otimes \underbrace{b_j}_{j\text{-th}} \otimes \ldots \otimes I_{\mathbb{C}^d}
        \end{equation}
        for some $b_j \in \mathbb{C}^{d \times d}$. This is a special case of our model~\eqref{eq:model}. All examples introduced in Section~\ref{sec:examples} will fall in this class. 
        
        With this choice, (i) in Assumption~\ref{as:Bj} becomes equivalent to $b_j^* b_j = b_j b_j^*$ and (ii) is trivially satisfied. Point (iii) is fulfilled, if and only if at least one of the matrices $b_1, \ldots, b_N$ has full rank. To see the latter statement, first observe that if some $b_j$ has full rank, then $\ker b_j = \{0\} \Rightarrow \ker B_j = \{0\} \Rightarrow \bigcap_j \ker B_j=\{0\}$. Conversely, if no $b_j$ has full rank, then we can find one $0 \neq u^{(j)} \in \mathbb{C}^d$ for each $b_j$, such that $b_j u^{(j)} = 0$. But then, $0 \neq u^{(1)} \otimes \ldots \otimes u^{(N)} \in \bigcap_j \ker B_j=\{0\} $.
   \end{remark}

    \begin{remark}[Infinite-dimensional $\mathfrak{h}$]\label{rem:infdim}
        It is not too difficult to see that our result also holds true if $\mathfrak{h}$ is infinite-dimensional and $B_j$ are compact operators satisfying Assumption~\ref{as:Bj}. In fact, the properties of $B_j$ only enter in Propositions~\ref{prop:kerA} and~\ref{prop:ranA}, as well as in Eq.~\eqref{eq:Anbound}. In the latter equation, the boundedness of all $B_j$ is exploited to prove that certain operators are closable\footnote{Here, we are adopting the following nomenclature: An operator $A$ on a Hilbert space $\hfrak$, with domain $\dom(A)$, is called closable if and only if there exists a closed operator $B: \dom(B)\subseteq\hfrak \to \hfrak, \dom(B) \supseteq \dom(A)$ that extends $A$ (see e.g.~\cite[Definition 1.3]{schmudgen2012unbounded}).}. Proposition~\ref{prop:kerA} also holds true for $B_j$ being generic bounded operators. In Proposition~\ref{prop:ranA}, we construct a common eigenbasis of $\{B_j\}_{j=1}^N$, which would also be possible if $B_j$ were compact operators. It is easy to see that in this case, the proof of Proposition~\ref{prop:ranA} (with an infinite eigenbasis) would still go through.
        
        In the case in which all $B_j$ are bounded and satisfy Assumption~\ref{as:Bj}, the proof of Proposition~\ref{prop:ranA} has to be modified using spectral calculus in order to obtain an analogue of Theorem~\ref{thm:main}. This generalization is straightforward but technical. To keep the current article short, we postpone it to a future publication.
    \end{remark}
    
    \section{Some examples}\label{sec:examples}

    Before starting with the proof of Theorem~\ref{thm:main}, let us discuss some particular instances of Hamiltonians in the form~\eqref{eq:model} to which Theorem~\ref{thm:main} applies.\medskip

    \paragraph*{\textbf{The $\sigma_x$ spin--boson model.}} Consider the following expression on $\mathbb{C}^2\otimes\fock$:
    \begin{equation}\label{eq:sigmax}
        H=\frac{\eta}{2}\,\sigma_z\otimes I_{\fock}+I_{\mathbb{C}^2}\otimes\dOmega+\sigma_x\otimes\left(\adag{f}+\a{f}\right),
    \end{equation}
    where $f\in\hilb_{-s}$, $s\leq 1$, and $\eta\geq0$; furthermore, $\sigma_x$ and $\sigma_z$ denote the first and third Pauli matrices, respectively:
    \begin{equation}
        \sigma_x=
        \begin{bmatrix}
            0&1\\
            1&0
        \end{bmatrix},\qquad    
        \sigma_z=
        \begin{bmatrix}
            1&0\\0&-1
        \end{bmatrix}.
    \end{equation}
    This is the model usually referred to as ``the'' spin--boson model in the mathematical literature, cf.~\cite{hirokawa2001remarks,hubner1995spectral,hirokawa1999expression,arai1990asymptotic,amann1991ground,davies1981symmetry,fannes1988equilibrium,hubner1995radiative}. Physically, it describes a two-level system (qubit), whose energy levels differ by a quantity $\eta$, undergoing \textit{decay} dynamics as a result of its interaction with the boson field; for example, it is employed in quantum optics to describe the interaction of a two-level atom with an electromagnetic field. In the monochromatic case ($\hilb=\mathbb{C}$), it reduces to the well-known Rabi model~\cite{xie2017quantum,braak2011integrability,hwang2015quantum,zhong2013analytical}.

    This model clearly belongs to the class of operators analyzed in the present work, corresponding to the case $D=2$, $N=1$, and with $K=\eta/2\;\sigma_z$, $B=\sigma_x$ (thus $A=\sigma_x\otimes\a{f}$). As such, Theorem~\ref{thm:main} applies. Using the obvious isomorphism $\hfrak=\mathbb{C}^2\otimes\fock\simeq\fock\oplus\fock$---that is, representing the most general element of $\hfrak$ as
    \begin{equation}
        \Psi=\begin{bmatrix}
        \Psi_\e\\\Psi_\g
        \end{bmatrix}:\quad \Psi_\e,\Psi_\g\in\fock,
    \end{equation}
    with $\Psi_\e$ and $\Psi_\g$ respectively playing the role of the ``excited'' and ``ground'' wavefunction of the boson field, we can conveniently rewrite Eq.~\eqref{eq:sigmax} as
    \begin{eqnarray}\label{eq:sigmax_explicit}
        H&=&\begin{bmatrix}
            \dOmega+\frac{\eta}{2}  &   \adag{f}+\a{f}         \\
            \adag{f}+\a{f}          & \dOmega-\frac{\eta}{2}
        \end{bmatrix}\nonumber\\
        &=&\underbrace{
        \begin{bmatrix}
            \dOmega+\frac{\eta}{2}&                            \\
                                  & \dOmega-\frac{\eta}{2}
        \end{bmatrix}}_{H_{\textnormal{free}}}+
        \underbrace{
        \begin{bmatrix}
                    &    \adag{f}   \\
         \adag{f}   & 
        \end{bmatrix}}_{A^*}+
        \underbrace{
        \begin{bmatrix}
                    &   \a{f}   \\
          \a{f}     & 
        \end{bmatrix}}_{A}.
    \end{eqnarray}
    Clearly, Assumption~\ref{as:Bj} is satisfied since $\sigma_x$ is Hermitian (thus normal) and invertible. By Theorem~\ref{thm:main}, the expression above defines a self-adjoint operator on the domain
    \begin{equation}\label{eq:sigmax_domain}
        \dom(H)=\left\{\begin{bmatrix}
        \Psi_\e\\\Psi_\g
        \end{bmatrix}:\;\begin{bmatrix}
        \Psi_\e+\frac{1}{\dOmega+\frac{\eta}{2}-z_0}\adag{f}\Psi_\g\\
        \Psi_\g+\frac{1}{\dOmega-\frac{\eta}{2}-z_0}\adag{f}\Psi_\e
        \end{bmatrix}\in\dom(H_{\textnormal{free}})   
        \right\},
    \end{equation}
    with $\dom(H_{\textnormal{free}})=\dom(\dOmega)\oplus\dom(\dOmega)$ in this representation. Here, $z_0$ is any fixed real number in the resolvent set of $H_{\textnormal{free}}$. 

    This expression should be compared with the one in~\cite[Sections 5--6]{lonigro2022generalized}, where the \emph{rotating-wave approximation} (RWA) spin--boson model was investigated. The RWA corresponds (cf.~\cite{agarwal1971rotating}) to replacing $\sigma_x\otimes(\a{f}+\adag{f})$ by $\sigma_+\otimes\a{f}+\sigma_-\otimes\a{f}$, where
    \begin{equation}
        \sigma_+ = \begin{bmatrix} 0&1 \\ 0&0 \end{bmatrix}=\sigma_-^*,\qquad \sigma_x=\sigma_++\sigma_-;
    \end{equation}
    physically, counter-rotating terms are removed. The expression of the domain without suppressing such terms is indeed more involved---the ``excited'' and ``ground'' components of the total state acquire a \textit{mutual} interdependence, differently from what happens in the rotating-wave case in which only the ground component depends on the excited one. Note that, since $\sigma_+$ has a nontrivial kernel and is not normal, Theorem~\ref{thm:main} does not apply to the RWA spin--boson model; nevertheless, the self-adjointness domain provided in~\cite{lonigro2022generalized} is in fact compatible with Eq.~\eqref{eq:Hdomain}, cf.~\cite[Section 6]{lonigro2022renormalization}. We briefly comment on this point in Section~\ref{sec:conclusions}.\medskip

    \paragraph*{\textbf{Multi-atom generalization.}} The discussion above can be immediately generalized to the case of $N$ atoms, cf. Remark~\ref{rem:spinspace}. Given $N\in\mathbb{N}$, we consider the following expression on $\mathbb{C}^{2^N}\otimes\fock$:
    \begin{equation}\label{eq:sigmax_multi}
        H=\sum_{j=1}^N\frac{\eta_j}{2}\sigma_{z,j}\otimes I_{\fock}+I_{\mathbb{C}^{2^N}}\otimes\dOmega+\sum_{j=1}^N\sigma_{x,j}\otimes\left(\adag{f_j}+\a{f_j}\right),
    \end{equation}
    where $f_1,\dots,f_N$ constitute an $\hilb$-independent family of form factors in $\hilb_{-s}$, $s\leq1$, and $\eta_1,\dots,\eta_N\geq0$. The operators $\sigma_{x,j}$, $\sigma_{z,j}$ are defined as in Eq.~\eqref{eq:identical}, i.e.,
    \begin{equation}
        \sigma_{x,j}=I_{\mathbb{C}^2}\otimes\ldots\otimes\underbrace{\sigma_x}_{j\text{-th}}\otimes\ldots I_{\mathbb{C}^2},\qquad  \sigma_{z,j}=I_{\mathbb{C}^2}\otimes\ldots\otimes\underbrace{\sigma_z}_{j\text{-th}}\otimes\ldots I_{\mathbb{C}^2}.
    \end{equation}
    Again, with this choice of parameters, Assumption~\ref{as:Bj} is satisfied (cf. Remark~\ref{rem:spinspace}) and Theorem~\ref{thm:main} applies. 

    Similarly, as in the case $N=1$, one may find explicit representations of the action of $H$ in matrix form, as well as its self-adjointness domain $\dom H$ in vector form, by exploiting the obvious isomorphism 
    \begin{equation}
        \hfrak=\mathbb{C}^{2^N}\otimes\fock\simeq\overbrace{\fock\oplus\cdots\oplus\fock}^{2^N\text{ times }}.
    \end{equation}
    This is a natural choice (though possibly not the most natural one) obtained by representing the tensor product in $\mathbb{C}^{2^N}$ as a Kronecker product. As an example, we shall present the case $N=2$ and, for the ease of exposition, we shall set $\eta_1=\eta_2\equiv\eta$. The most general element of $\hfrak$ is represented as
    \begin{equation}
        \Psi=\begin{bmatrix}
        \Psi_{\e\e}\\\Psi_{\e\g}\\\Psi_{\g\e}\\\Psi_{\g\g}
        \end{bmatrix},\qquad\Psi_{\e\e},\Psi_{\e\g},\Psi_{\g\e},\Psi_{\g\g} \in \fock,
    \end{equation}
    and, in this representation, Eq.~\eqref{eq:sigmax_multi} reads\small
    \begin{eqnarray}\label{eq:sigmax_multi_explicit}
        H&=&
        \begin{bmatrix}
            \dOmega+\eta        &  \adag{f_2}+\a{f_2}     &  \adag{f_1}+\a{f_1}     &                         \\
            \adag{f_2}+\a{f_2}  &    \dOmega              &                         &   \adag{f_1}+\a{f_1}    \\
            \adag{f_1}+\a{f_1}  &                         &   \dOmega               &   \adag{f_2}+\a{f_2}    \\
                                &   \adag{f_1}+\a{f_1}    &    \adag{f_2}+\a{f_2}   &     \dOmega-\eta  
        \end{bmatrix}\nonumber\\
         &=&\underbrace{
        \begin{bmatrix}
            \dOmega+\eta    &               &              &                    \\
                            &    \dOmega    &              &                    \\
                            &               &   \dOmega    &                    \\
                            &               &              &     \dOmega-\eta  
        \end{bmatrix}}_{H_{\textnormal{free}}}
        \nonumber\\&&+
        \underbrace{
        \begin{bmatrix}
                         &   \adag{f_2}   &  \adag{f_1}   &                 \\
            \adag{f_2}   &                &               &   \adag{f_1}    \\
            \adag{f_1}   &                &               &   \adag{f_2}    \\
                         &   \adag{f_1}   &  \adag{f_2}   &       
        \end{bmatrix}}_{A^*}
            +
        \underbrace{
        \begin{bmatrix}
                      &  \a{f_2}   &  \a{f_1}   &              \\
            \a{f_2}   &            &            &   \a{f_1}    \\
            \a{f_1}   &            &            &   \a{f_2}    \\
                      &   \a{f_1}  &  \a{f_2}   &       
        \end{bmatrix}}_{A},
    \end{eqnarray}\normalsize
    from which one can reconstruct the analogue of Eq.~\eqref{eq:sigmax_domain} for $N=2$ atoms.

    We remark that, while we have only considered here the case in which $H_{\textnormal{free}}$ does not involve off-diagonal terms (spin--spin interactions), Theorem~\ref{thm:main} also covers this possibility---of course, in such a case all expressions, e.g., Eq.~\eqref{eq:sigmax_multi_explicit} in the case $N=2$, will acquire additional terms and become more involved. Finally, like in the single-atom case, the RWA of this model with mild divergences (which is not covered by Theorem~\ref{thm:main}) was covered in~\cite{lonigro2023self}.\medskip

    \paragraph*{\textbf{The $\sigma_z$ (dephasing) spin--boson model.}}

    Let us consider the following expression:
    \begin{equation}\label{eq:sigmaz}
        H=\frac{\eta}{2}\,\sigma_z\otimes I_{\fock}+I_{\mathbb{C}^2}\otimes\dOmega+\sigma_z\otimes\left(\adag{f}+\a{f}\right),
    \end{equation}
    again with $f\in\hilb_{-s}$, $s\leq 1$, and $\eta\geq0$. So the (off-diagonal) $\sigma_x$-interaction from above has been replaced with a (diagonal) $\sigma_z$-interaction. Physically, Eq.~\eqref{eq:sigmaz} describes a two-level system undergoing \textit{decoherence} as a result of its interaction with the boson field. Using the same representation as before, we can write it as
    \begin{eqnarray} \label{eq:Hsigmaz}
        H&=&
        \begin{bmatrix}
            \dOmega+\frac{\eta}{2}+\adag{f}+\a{f}&                                          \\
                                             & \dOmega-\frac{\eta}{2}+\adag{f}+\a{f}
        \end{bmatrix}\nonumber\\
        &=&\underbrace{
        \begin{bmatrix}
        \dOmega+\frac{\eta}{2}      &                           \\
                                    & \dOmega-\frac{\eta}{2}
        \end{bmatrix}}_{H_{\textnormal{free}}}+
        \underbrace{
        \begin{bmatrix}
            \adag{f}    &           \\
                        & \adag{f}
        \end{bmatrix}}_{A^*}+
        \underbrace{
        \begin{bmatrix}
            \a{f}   &           \\
                    & \a{f}
    \end{bmatrix}}_{A}.
    \end{eqnarray}
    By Theorem~\ref{thm:main}, this expression defines a self-adjoint operator on the domain
    \begin{equation}
    \dom(H)=\left\{\begin{bmatrix}
        \Psi_\e\\\Psi_\g
    \end{bmatrix}:\;\begin{bmatrix}
       \bigl(1+\frac{1}{\dOmega+\frac{\eta}{2}-z_0}\adag{f}\bigr)\Psi_\e\\
        \bigl(1+\frac{1}{\dOmega-\frac{\eta}{2}-z_0}\adag{f}\bigr)\Psi_\g
    \end{bmatrix}\in\dom(H_{\textnormal{free}})   
    \right\},
    \end{equation}
    which, differently from the $\sigma_x$-case, involves no mutual dependence between the two components $\Psi_\e,\Psi_\g$, due to the inherently diagonal nature of the interaction. In particular, Eq.~\eqref{eq:Hsigmaz} can be regarded as a direct sum of two van Hove models with different energy offsets $\pm \frac{\eta}{2}$.

    Similar expressions can be found for the multi-atom version of the model, following the same arguments presented for the $\sigma_x$ model---however, while in the previous case the structure of the multi-atom model for $N>1$ becomes quickly involved, this model will preserve a relatively simple, diagonal structure. For instance, in the case $N=2$ (and again setting $\eta_1=\eta_2\equiv\eta$), one has\small
    \begin{eqnarray}
        H&=&\underbrace{
        \begin{bmatrix}
            \dOmega+\eta  &           &           &                 \\
                          &  \dOmega  &           &                 \\
                          &           &  \dOmega  &                 \\
                          &           &           &   \dOmega-\eta  
        \end{bmatrix}}_{H_{\textnormal{free}}}
        \nonumber\\&&+
        \underbrace{
        \begin{bmatrix}
            \adag{f_+}  &              &              &              \\
                        &  \adag{f_-}  &              &              \\
                        &              & -\adag{f_-}  &              \\
                        &              &              &  -\adag{f_+}   
        \end{bmatrix}}_{A^*}
         +
        \underbrace{
        \begin{bmatrix}
            \a{f_+}  &           &           &           \\
                     &  \a{f_-}  &           &           \\
                     &           & -\a{f_-}  &           \\
                     &           &           &  -\a{f+}   
    \end{bmatrix}}_{A},
    \end{eqnarray}\normalsize
    where $f_\pm:=f_1\pm f_2$. Of course, the $\hilb$-independence of the form factors $f_1$ and $f_2$ ensures that $f_\pm \in \hilb_{-1}\setminus\hilb$, so both combinations are UV-divergent.  
	
    \section{Self-adjointness of \texorpdfstring{$H_{\textnormal{free}}+A^*+A$}{H+A*+A}}\label{sec:posilicano}

    \subsection{Generalized setting for \texorpdfstring{$H_{\textnormal{free}}+A^*+A$}{H+A*+A}}
    
	We shall start by revising the main results obtained in~\cite{posilicano2020self} (with a slightly modified notation). These results hold for  $\hfrak$ being any Hilbert space, and $H_{\textnormal{free}}$ any densely defined, self-adjoint operator on $\hfrak$ with domain $\dom (H_{\textnormal{free}})\subset\hfrak$. Besides, we will assume $H_{\textnormal{free}}$ to be bounded from below; to simplify the notation, we will then assume $H_{\textnormal{free}}\geq0$, the general case being recovered by a simple shift. We can then define the Hilbert scale $\{\hfrak_s\}_{s\in\mathbb{R}}$ associated with $H_{\textnormal{free}}$ in the usual way: for any $s\in\mathbb{R}$,
	\begin{equation}\label{eq:scale}
		\hfrak_s:=\overline{\bigcap_{n\in\mathbb{N}}\dom (H_{\textnormal{free}}^n)}^{\|\cdot\|_{\hfrak_s}},\qquad \left\|\Psi\right\|_{\hfrak^s}:=\left\|(H_{\textnormal{free}}+1)^{s/2}\Psi\right\|_\hfrak,
	\end{equation}
	with $\hfrak_{+s}\subset\hfrak_0=\hfrak\subset\hfrak_{-s}$ for any $s\geq0$, and the spaces $\hfrak_{\pm s}$ being naturally dual. Notice that $\hfrak_{+s}=\dom (H_{\textnormal{free}}^{s/2})$ for any $s\geq0$. With a slight abuse of notation, the duality pairing between the two spaces $\hfrak_{\pm s}$ will sometimes be denoted like the scalar product on $\hfrak$ by $\Braket{\cdot,\cdot}$ or $\Braket{\cdot,\cdot}_{\fock_{-s},\fock_s}$. Notice that the inclusions
	\begin{equation}
		\mathcal{B}(\hfrak_{+s},\hfrak_{-s})\subset\mathcal{B}(\hfrak_{+s'},\hfrak_{-s'})
	\end{equation}
	hold for all $s'>s>0$. With this formalism, $H_{\textnormal{free}}$ can be interpreted as a continuous operator from $\hfrak_{s}$ to $\hfrak_{s-2}$ for any $s\in\mathbb{R}$, i.e., it decreases the regularity by 2. Now, suppose that for some $s' \ge 1$, we have $A\in\mathcal{B}(\hfrak_{+s'},\hfrak)$ and $A^*\in\mathcal{B}(\hfrak,\hfrak_{-s'})$. These two statements are expected to hold if $A^*, A$ involve linear couplings to a boson field with form factors $f_j \in \hilb_{-s'}$ (in our case, $f_j \in \hilb_{-s}, s \le 1 \Rightarrow f_j \in \hilb_{-1}$). If they do hold, then the following operator is already well-defined and continuous:
	\begin{equation}\label{eq:model2continuous}
		(H_{\textnormal{free}}+A^*+A):\hfrak_{+s'}\rightarrow\hfrak_{-s'}.
	\end{equation}
    However, in order to obtain quantum dynamics, we need an operator defined on a dense subspace of $\hfrak$ with values in $\hfrak$ (that is, it ``does not leave'' the Hilbert space $\hfrak$) and which is self-adjoint. It turns out that in our less singular Case~1 (so $s \in (0,1]$), this can be achieved by restricting $H_{\textnormal{free}}+A^*+A$ to a suitable domain $\dom (H)$. The more singular Case~2 (so $s\in(1,2]$), is expected to require a further modification of Eq.~\eqref{eq:model2continuous}, which is formally equivalent to subtracting an infinite self-energy. We comment on this in Remark~\ref{rem:selfenergy}.

    The domain $\dom (H)$ given by~\cite{posilicano2020self} is conveniently described in terms of the resolvents
	\begin{equation}\label{eq:Rz}
		R_z:=(H_{\textnormal{free}}-z)^{-1},\qquad z\in\rho(H_{\textnormal{free}})\supset\mathbb{C}.
	\end{equation}
    To simplify notation, as in~\cite{posilicano2020self}, we fix some $z_0\in\rho(H_{\textnormal{free}})\cap\mathbb{R}$ and set 
    \begin{equation}
       R:=R_{z_0}.
    \end{equation}
    Since $H_{\textnormal{free}}\ge0$, we can choose $z_0=-1$ without loss of generality. All resolvents improve regularity by 2, in the sense that they can be interpreted as continuous operators $R_z, R: \hfrak_s \to \hfrak_{s+2}$ for any $s\in\mathbb{R}$. Furthermore, the resolvent vanishes as $z$ moves away from the spectrum in the following sense.
    
    \begin{lemma}\label{lem:resolventvanishing}
        Let $s\in[0,2), \{z_n\}_{n \in \mathbb{N}} \subset \mathbb{C}, \Im z_n \le 0$ such that $\lim_{n\to\infty} \mathrm{dist}(z_n, \sigma(H_{\textnormal{free}})) = \infty$. Then, for any $r\in\mathbb{R}$, we have
        \begin{equation}\label{eq:resolventvanishing}
            \lim_{n \to \infty} \Vert R_{z_n} \Vert_{\hfrak_{r-s} \to \hfrak_r} = 0.
        \end{equation}
    \end{lemma}
    
    \begin{proof}
    We first consider the case $r=0$. By spectral calculus, for $\Psi \in \hfrak_{-s}$ we have
    \begin{equation}
    \begin{aligned}
        \Vert R_{z_n} \Psi \Vert_{\hfrak}
        = &\Vert (H_{\textnormal{free}} - z_n)^{-1} \Psi \Vert_{\hfrak}
        \le \left\Vert \frac{(H_{\textnormal{free}} + 1)^{\frac s2}}{H_{\textnormal{free}} - z_n} \right\Vert_{\hfrak \to \hfrak} \Vert (H_{\textnormal{free}} + 1)^{-\frac s2} \Psi \Vert_{\hfrak}\\
        \le &\sup_{\lambda \in \sigma(H_{\textnormal{free}})} \left\vert \frac{(\lambda + 1)^{\frac s2}}{\lambda - z_n} \right\vert \Vert \Psi \Vert_{\hfrak_{-s}}
        \le \sup_\lambda \left\vert \frac{\lambda + 1}{\lambda - z_n} \right\vert^{\frac s2} |\lambda - z_n|^{-1 + \frac s2} \Vert \Psi \Vert_{\hfrak_{-s}}\\
        \le &\Big( \underbrace{\sup_\lambda \frac{\lambda}{|\lambda-z_n|}}_{\le 1} + \underbrace{\sup_\lambda \frac{1}{|\lambda-z_n|}}_{\to 0} \Big)^{\frac s2} \underbrace{\mathrm{dist}(z_n, \sigma(H_{\textnormal{free}}))^{-1 + \frac s2}}_{\to 0} \Vert \Psi \Vert_{\hfrak_{-s}}.
    \end{aligned}
    \end{equation}
    The statement for general $r \in \mathbb{R}$  immediately follows, as for any $\Phi \in \hfrak_{r-s}$ and $\Psi := (H_{\textnormal{free}} + 1)^{\frac r2} \Phi \in \hfrak_{-s}$, we have 
    \begin{equation}
        \Vert R_{z_n} \Phi \Vert_{\hfrak_r}
        = \Vert R_{z_n} (H_{\textnormal{free}} + 1)^{-\frac r2} \Psi \Vert_{\hfrak_r}
        = \Vert (H_{\textnormal{free}} + 1)^{-\frac r2} R_{z_n} \Psi \Vert_{\hfrak_r}
        = \Vert R_{z_n} \Psi \Vert_{\hfrak},
    \end{equation}
    and $\Vert \Psi \Vert_{\hfrak_{-s}} = \Vert \Phi \Vert_{\hfrak_{r-s}}$.
    \end{proof}
    
	Now, $\dom (H) $ as in~\cite{posilicano2020self} (also cf.~\cite{lampart2018particle}--\cite{lampart2019nelson}) involves vectors of the form
	\begin{equation}\label{eq:singular}
		\Psi\in\hfrak:\;\;(I_{\hfrak}+RA^*)\Psi=\Psi_0\in\dom (H_{\textnormal{free}}).
	\end{equation}
    The results of~\cite{posilicano2020self} are valid in Cases 1 and 2, viz. $s\in(0,2]$. Now, a direct check shows that a vector $\Psi$ as in Eq.~\eqref{eq:singular} does \textit{not} belong to $\hfrak_{+s}$ whenever $s>1$, so that the action of $A$ on $\Psi$ is ill-defined for $s\in(1,2]$. For that reason, in~\cite{posilicano2020self} a modification $A_T$ of $A$ is introduced, which is indexed by a symmetric operator $T:\dom (T)\subset\hfrak\rightarrow\hfrak$ and defined as
    \begin{eqnarray}\label{eq:at1}
		\dom (A_T)&=&\left\{\Psi\in\hfrak:\;\exists\Phi\in\dom (T)\;\text{s.t. }\Psi+RA^*\Phi\in\dom (H_{\textnormal{free}})\right\},\\
		A_T\Psi&=&A\left(\Psi+RA^*\Phi\right)+T\Phi.
	\end{eqnarray}
	Since $\Psi+RA^*\Phi\in\dom (H_{\textnormal{free}})=\hfrak_{+2}\subset\hfrak_{+s}$, the operator $A_T$ is well-defined also for $1<s\leq2$ (in contrast to $A$). We conclude that the modified Hamiltonian
	\begin{eqnarray}\label{eq:model2modified}
        \dom (H_T)&=&\left\{\Psi\in\dom (T):(I_{\hfrak}+RA^*)\Psi=\Psi_0\in\dom (H_{\textnormal{free}})\right\} \subset \dom(A_T)\\\label{eq:model2modified2}
        H_T&=&H+A^*+A_T
	\end{eqnarray}
	is also well-defined. In particular, on $\dom (H_T)$, the operator $A_T$ defined in~\cite[Sect.~3]{posilicano2020self} acts as\footnote{We remark that the definition of $A_T$ in~\cite[Sect.~3]{posilicano2020self} requires a unique splitting $\Psi = \Psi_0 - R A^* \Phi$ with $\Psi_0\in \hfrak_2$. In case of our multi-spin--boson model~\eqref{eq:model}, this splitting is only unique if $s>0$. Otherwise, $R A^*$ maps $\hfrak \to \hfrak_2$, and we can arbitrarily shift vectors between $\Psi_0$ and $R A^* \Phi$. So the machinery of~\cite{posilicano2020self} strictly applies to Cases 1 and 2 ($s \in (0,2]$) and not to the less singular Case~0 ($s\le0$).}
	\begin{equation}\label{eq:at2}
		A_T\Psi
        =A\left(\Psi+RA^*\Psi\right)+T\Psi
        =A\Psi_0+T\Psi.
	\end{equation}
	As~\cite{posilicano2020self} makes statements about $H_T$, we need to find out which $T$ makes $H_T$ correspond to $H = H_{\textnormal{free}}+A^*+A$ in case $s\in(0,1]$. A short calculation reveals
    \begin{equation}
        H_0 \Psi
        = (H_{\textnormal{free}} + A^*) \Psi + A \Psi_0
        = (H_{\textnormal{free}} + A^*) \Psi + A \Psi + ARA^* \Psi.
    \end{equation}
    Thus, to obtain a renormalized Hamiltonian $H_T$ that physically corresponds to our model Hamiltonian $H_{\textnormal{free}} + A^* + A$, we need to ``manually add'' the interaction
	\begin{equation}\label{eq:T}
		T = -ARA^*
	\end{equation}
	to $H_0$. In fact, for $s\leq1$ we have $ARA^*\in\mathcal{B}(\hfrak)$, whereas $ARA^*$ is ill-defined for $s>1$, since we have the mappings and inclusions
	\begin{equation}\label{eq:Tbound}
		\hfrak\xrightarrow{A^*}\hfrak_{-s}\xrightarrow{R}\hfrak_{-s+2}\begin{cases}
			\subset\hfrak_{+s}\xrightarrow{A}\hfrak,&0\leq s\leq1;\\
			\supset\hfrak_{+s}&1< s\leq2.
		\end{cases}
	\end{equation}
    In case $s>1$, a different choice of $T$, involving an infinite self-energy, would become necessary, as we explain in Remark~\ref{rem:selfenergy}.

    \subsection{Self-adjointness and convergence results}
    
    In~\cite{posilicano2020self}, it is now shown that, under suitable assumptions on $A$, the operator $H_T$ is self-adjoint on $\dom (H_T)$, as defined in Eqs.~\eqref{eq:model2modified}--\eqref{eq:model2modified2}, for a large class of choices of $T$ (including all bounded operators). Furthermore,~\cite{posilicano2020self} provides conditions for when $H_T$ is the norm resolvent limit of some regularized operators $H_n:=H_{\textnormal{free}}+A_n^*+A_n$, as they are used in cutoff renormalization~\cite{nelson1964interaction, eckmann1970model, frohlich1973infrared, sloan1974polaron, griesemer2016self, griesemer2018domain}. In the following, we summarize some of these results. Let us introduce the following notation: for every $z\in\rho(H_{\textnormal{free}})$, we define $M(z)$ as the bounded operator on $\hfrak$ given by
	\begin{equation}
		M(z):=A(R_z-R)A^*.
	\end{equation}
	This is indeed a bounded operator on $\hfrak$ whenever $A\in\mathcal{B}(\hfrak_2,\hfrak)$ (or, a fortiori, $A\in\mathcal{B}(\hfrak_s,\hfrak)$ for some $s\in(0,2]$) since, by the second resolvent identity, the difference between the resolvents of $H_{\textnormal{free}}$, evaluated in two distinct points of its resolvent, maps $\hfrak_{-2}$ into $\hfrak_{+2}$. With this notation, combining~\cite[Lemma~3.6]{posilicano2020self} and~\cite[Corollary~3.7]{posilicano2020self}, we now get the following result.
	
	\begin{theorem}\label{thm:posilicano1}
	Let $H_{\textnormal{free}}$ be a self-adjoint operator, $A\in\mathcal{B}(\hfrak_s,\hfrak)$ for some $s\in(0,2)$, with $A^*\in\mathcal{B}(\hfrak,\hfrak_{-s})$ its dual, and $T:\dom (T)\subset\hfrak\rightarrow\hfrak$ self-adjoint. Suppose that
	\begin{itemize}
		\item[(i)] $\ker A|_{\hfrak_2}$ and $\ran A|_{\hfrak_2}$ are dense in $\hfrak$;
		\item[(ii)] there exists $z\in\rho(H_{\textnormal{free}})$ such that $M(z)-T$ has a bounded inverse;
		\item[(iii)] there exists $t\in[0,1-\frac s2)$ such that $\dom (T)\supset\hfrak_{2t}$ and $T|_{\hfrak_{2t}}\in\mathcal{B}(\hfrak_t,\hfrak)$.
	\end{itemize}
    Then, the following statement holds true:
    \begin{itemize}
        \item[(a')] The operator $H_T$ defined by
        \begin{eqnarray}\label{eq:domht}
            \dom (H_T)&=&\left\{\Psi\in\hfrak_{2-s}:\left(I_{\hfrak}+RA^*\right)\Psi\in\dom (H_{\textnormal{free}})\right\};\\\label{eq:ht}
            H_T&=&H_{\textnormal{free}}+A^*+A_T
        \end{eqnarray}
        is self-adjoint, and its resolvent is given by~\cite[(3.11)]{posilicano2020self}:
        \begin{equation}\label{eq:resolventformula}
            (H_T-z)^{-1} = (H_{\textnormal{free}}-z)^{-1} -
            \begin{bmatrix}
                R_z A^* & R_z
            \end{bmatrix}
            \begin{bmatrix}
                A_T R_z A^* & A R_z + 1\\
                R_z A^* + 1 & R_z
            \end{bmatrix}^{-1}
            \begin{bmatrix}
                A R_z\\
                R_z
            \end{bmatrix}.
        \end{equation}
    \end{itemize}
	\end{theorem}
    Note that requirement (ii) is automatically satisfied whenever $T$ itself has a bounded inverse.
    
    \begin{proof}
   ~\cite[Corollary~3.7]{posilicano2020self} implies that (a') holds if we can verify (i) and the assumptions of~\cite[Lemma~3.6]{posilicano2020self}. The latter lemma, in turn, assumes that (ii) and (iii) hold, as well as
    $\exists \gamma \in \mathbb{R}$ with $|\gamma|$ large enough, such that $(1 - G_{\pm i \gamma})^{-1}$ and $(1 - G_{\pm i \gamma}^*)^{-1}$ are bounded operators on $\hfrak$, where $G_z := - R_z A^*$. But now,
    \begin{equation}
        \Vert G_{\pm i \gamma} \Vert_{\hfrak \to \hfrak} \le \Vert A^* \Vert_{\hfrak \to \hfrak_{-s}} \Vert R_{\pm i \gamma} \Vert_{\hfrak_{-s} \to \hfrak}.
    \end{equation}
    By Lemma~\ref{lem:resolventvanishing}, this norm becomes arbitrarily small as $|\gamma| \to \infty$, eventually dropping below 1, so $(1 - G_{\pm i \gamma})$ has a bounded inverse for large enough $|\gamma|$. The same reasoning applies to $(1 - G_{\pm i \gamma}^*)$. So the additional invertibility condition is automatically satisfied.
    \end{proof}
    
    Although the above conditions for self-adjointness of $H_T$ are already very useful,~\cite{posilicano2020self} provides another incarnation of its main result, which applies even more conveniently to the case of bounded $T$.
    
    \begin{theorem}[{\cite[Theorem 3.13]{posilicano2020self}}]\label{thm:posilicano2}
        Let $H_{\textnormal{free}}$ be a self-adjoint nonnegative operator, $A\in\mathcal{B}(\hfrak_{+s},\hfrak)$ for some $s\in(0,2)$, with $A^*\in\mathcal{B}(\hfrak,\hfrak_{-s})$ its dual, and $T:\dom (T)\subset\hfrak\rightarrow\hfrak$ self-adjoint. Suppose that
	    \begin{itemize}
	          \item[(i)] $\ker A|_{\hfrak_2}$ and $\ran A|_{\hfrak_2}$ are dense in $\hfrak$.
	    \end{itemize}
        Then, (a') in Theorem~\ref{thm:posilicano1} for $T=0$ holds true, that is, the operator $H_0$, defined by Eqs.~\eqref{eq:domht}--\eqref{eq:ht} with $T=0$, is self-adjoint. Furthermore, if some (possibly nonzero) $T:\dom (T)\subset\hfrak\rightarrow\hfrak$, symmetric is given, such that
        \begin{itemize}
	        \item[(ii)] $T$ is $H_0$-bounded,\footnote{By \textit{$X$ is $Y$-bounded} or \textit{$X$ is relatively bounded with respect to $Y$} we mean $\exists a, b \ge 0 $ s.t. $\Vert X \Psi \Vert \le a \Vert Y \Psi \Vert + b \Vert \Psi \Vert \; \forall \Psi \in \dom (Y)$. The constant $a$ is called \textit{relative bound}. 
            %For $a<1$, we call $X$ \textit{Kato-bounded w.r. to $Y$}. If for every $a>0$, there exists some $b$ such that the relative bound is true, we call $X$ \textit{infinitesimally $Y$-bounded}.
            Sometimes, when $a<1$, $X$ is said to be \textit{Kato-bounded w.r. to $Y$}. If for every $a>0$, there exists some $b$ such that the relative bound is true, $X$ is said to be \textit{infinitesimally $Y$-bounded}.}
        \end{itemize}
        then, (a') in Theorem~\ref{thm:posilicano1} holds true also for that $T$, that is, the operator $H_T$ defined by Eqs.~\eqref{eq:domht}--\eqref{eq:ht} is self-adjoint, and its resolvent is given by Eq.~\eqref{eq:resolventformula}.
    \end{theorem}
    
    Finally, the question remains, if the renormalized Hamiltonian $H_T=H_{\textnormal{free}}+A^*+A$ is equivalent to a Hamiltonian obtained by cutoff renormalization. That is: does there exist a sequence of regularized Hamiltonians $H_n:=H_{\textnormal{free}}+A_n^*+A_n+E_n$ such that $H_n \to H_T$ in a suitable sense? Here,~\cite[Theorem~3.10]{posilicano2020self} provides us with conditions for an affirmative answer. The following result is an immediate consequence of~\cite[Theorem~3.10]{posilicano2020self}, using\footnote{Note that the consequences of Theorem~\ref{thm:posilicano2} allow us to restrict to the case $\mu=0$ in~\cite[Theorem~3.10]{posilicano2020self}.} Theorem~\ref{thm:posilicano2}.
	\begin{theorem}\label{thm:posilicano3}
		Let $H_{\textnormal{free}}$ a self-adjoint nonnegative operator, $A\in\mathcal{B}(\hfrak_{+s},\hfrak)$ for some $s\in(0,2)$, with $A^*\in\mathcal{B}(\hfrak,\hfrak_{-s})$ its dual, and $T:\dom (T)\subset\hfrak\rightarrow\hfrak$ self-adjoint. Suppose that
		\begin{itemize}
			\item[(i)] $\ker A|_{\hfrak_2}$ and $\ran A|_{\hfrak_2}$ are dense in $\hfrak$;
			\item[(ii)] $T$ is $H_0$-bounded with relative bound smaller than one;
		\end{itemize}
	besides, suppose that there exists a sequence of closable operators $A_n:\dom (A_n)\subset\hfrak\rightarrow\hfrak$ (``regularized annihilation operators''), such that
	\begin{itemize}
	    \item[(iii)] there exists some $s'\in[0,1]$ such that $\dom (A_n)\supset\hfrak_{s'}$ and\footnote{We remark that $ A_n $ is not symmetric, so we cannot conclude boundedness from the Hellinger--Toeplitz theorem, and have to impose it separately.} $A_n|_{\hfrak_{s'}}\in\mathcal{B}(\hfrak_{s'},\hfrak)$;
		\item[(iv)] $A_n^*+A_n$ is $H_{\textnormal{free}}$-bounded with $n$-independent relative bound smaller than one;
		\item[(v)] $\lim_{n\to\infty}\left\|A_n-A\right\|_{\hfrak_2 \to \hfrak}=0$;
		\item[(vi)] there exists a sequence $\{E_n\}_{n\in\mathbb{N}}\subset\mathcal{B}(\hfrak)$, $E_n$: symmetric, such that, for all $n\in\mathbb{N}$, $A_nRA_n^*+E_n$ is relatively bounded w.r.t. $H_{\textnormal{free}}+A_n^*+A_n+A_nRA_n^*$, with $n$-independent relative bound smaller than one, and
	    \begin{equation}\label{eq:convergence}
		      \lim_{n\to\infty}\left\|A_nRA_n^*+E_n+T\right\|_{\dom (T) \to \hfrak}=0.
	    \end{equation}
	\end{itemize}
    Then, we have
    \begin{itemize}
        \item [(b')] $H_{\textnormal{free}}+A_n^*+A_n-E_n$ converges to $H_{\textnormal{free}}+A^*+A_T$ in the norm resolvent sense.
    \end{itemize}
	\end{theorem}
   
	\begin{remark}[On two different notions of adjoint]\label{rem:adjoints}
	Note that we are requiring each $A_n$ to be \textit{closable}, thus admitting an adjoint $A_n^\dag:\dom (A_n^\dag)\subset\hfrak\rightarrow\hfrak$ with respect to the topology of $\hfrak$. This is different from the adjoint with respect to the pairing between $\hfrak_{s'}$ and $\hfrak_{-s'}$ (for some $s'$ s.t. $\hfrak_{s'} \subset \dom (A_n)$), i.e., there are ``two different adjoints'':
	\begin{itemize}
		\item an unbounded operator $A_n^\dag:\dom (A_n^\dag)\subset\hfrak\rightarrow\hfrak$ defined by
		\begin{equation}
			\Braket{\Psi,A_n\Phi}_{\hfrak}=\braket{A_n^\dag\Psi,\Phi}_{\hfrak}\qquad\text{for all }\Phi\in\dom (A_n);
		\end{equation}
		\item a continuous operator $A_n^*:\hfrak\rightarrow\hfrak_{-s'}$ defined by
	\begin{equation}
		\Braket{\Psi,A_n\Phi}_{\hfrak}=\Braket{A_n^*\Psi,\Phi}_{\hfrak_{-s'},\hfrak_{s'}}\qquad\text{for all }\Phi\in\hfrak_{s'}, \Psi \in \hfrak.
	\end{equation}
	\end{itemize}
	Clearly $A_n^*$ is an extension of $A_n^\dag$ to all of $\hfrak$. With a slight abuse of notation, we will also write $A_n^*$ for the restriction of this operator to $\dom (A_n^\dag) \to \hfrak$ (i.e., for $A_n^\dag$). It then becomes clear from the context which domain is meant.
   	\end{remark}

    \begin{remark}[Self-energy renormalization] \label{rem:selfenergy}
    Theorem~\ref{thm:posilicano3} is valid for $s\in(0,2)$, i.e., both in Case~1 and most situations of Case~2. The latter case is the one in which the terms $E_n$ become necessary, which can physically be interpreted as ``self-energy counterterms''. As described above, to make the renormalized Hamiltonian $H_T$ formally correspond to the physically desired Hamiltonian $H_{\textnormal{free}} + A^* + A$, we would need to ``manually add'' the interaction $T=-ARA^*=\lim_{n\to\infty} -A_n R A_n^*$. Now, the sequence of operators $A_nRA_n^*$:
	\begin{itemize}
		\item will converge to $ARA^*$, if $s\leq1$ (Case 1);
		\item will generally not have a limit, if $s>1$ (Case 2),
	\end{itemize}
 where, in both cases, the convergence is to be understood in the sense of Eq.~\eqref{eq:convergence}.
    However, in the latter case, several QFT-type models admit a properly chosen family of (self-energy) counterterms $\{E_n\}_{n\in\mathbb{N}}$, such that the operator
    \begin{equation}
        T = \lim_{n \to \infty} -A_n R A_n^* - E_n
    \end{equation}
    is well-defined, within a suitable sense of convergence. For example, for the Nelson model~\cite{nelson1964interaction} in Case~2, $\{E_n\}_{n \in \mathbb{N}}$ just amounts to a divergent sequence of real numbers. Since adding a constant to the Hamiltonian does not change the Heisenberg equations of motion, the operator $T$ can heuristically be regarded as equivalent to the ill-defined expression $-ARA^*$. Likewise, $H_T$ is heuristically equivalent to the desired model Hamiltonian $H_{\textnormal{free}} + A^* + A$.
    
    Of course, the choice of $E_n$ is \textit{not unique}: different choices of $\{E_n\}_{n \in \mathbb{N}}$ may result in different operators $T$, and vice versa, various choices of $T$ are reachable by selecting an appropriate sequence of counterterms, resulting in inequivalent quantum dynamics. It is thus important to make a heuristic reasoning, based on physical arguments, to justify which renormalized Hamiltonian shall be the correct one.
    
    Instead, in our Case~1 ($s\in(0,1]$), we can simply set $E_n=0$ to obtain the limit $T=-ARA^*$ as required in Eq.~\eqref{eq:T}. The situation in Case~2 will be analyzed elsewhere.
    \end{remark}
    
   	\section{Proof of Theorem~\ref{thm:main}}

    Our plan is to use Theorems~\ref{thm:posilicano2} and~\ref{thm:posilicano3} to establish (a) and (b). We start by collecting some mathematical tools that will turn out to be useful for the proof.
    
	\subsection{Preliminaries}
	
	Let $\{\hilb_s\}_{s\in\mathbb{R}},\{\fock_s\}_{s\in\mathbb{R}}$ be the scales of Hilbert spaces associated respectively with $\omega$ and $\dOmega$, cf.~Eq.~\eqref{eq:hilbscale} and related discussion. Let us recall the following proposition.
 
    \begin{proposition}[{\cite[Props.~3.4, 3.5 and 3.7]{lonigro2022generalized}}]\label{prop:af}
	    Let $f\in\hilb_{-s}$ for some $s\geq1$. Then the following statements are true:
	    \begin{itemize}
		\item[(i)] the restriction of the annihilation operator to $\fock_{+s}$ defines a continuous operator $\a{f}\in\mathcal{B}(\fock_{+s},\fock)$;
		\item[(ii)] its adjoint $\adag{f}:=\a{f}^*\in\mathcal{B}(\fock,\fock_{-s})$ with respect to the duality pairing between $\fock_{+s}$ and $\fock_{-s}$ is a continuous operator whose action on $\Psi\in\fock_{+1}$ agrees with the creation operator defined in Eq.~\eqref{eq:adaggera};
		\item [(iii)] there exists a sequence $\{f_i\}_{i\in\mathbb{N}}\subset\hilb$ such that
		\begin{equation}
			\lim_{i\to\infty}\left\|a(f)-a(f_i)\right\|_{\fock_s \to \fock}=0,\qquad 	\lim_{i\to\infty}\left\|a^*(f)-a^*(f_i)\right\|_{\fock \to \fock_{-s}}=0,
		\end{equation}
		and this happens if and only if $\|f-f_i\|_{-s}\to0$.
	    \end{itemize}
    \end{proposition}
    In words: If $f\in\hilb_{-s}\setminus\hilb$, then the annihilation operator $\a{f}$, while still being a legitimate densely defined operator on $\fock$, fails to be closable and thus to have an adjoint with respect to the topology of $\hfrak$. However, its restriction to $\fock_{+s}$, interpreted as an operator in $\mathcal{B}(\fock_{+s},\fock)$, always admits an adjoint $\adag{f}\in\mathcal{B}(\fock,\fock_{-s})$ with respect to the pairing. If $f\in\hilb$, following the same line of reasoning as in Remark~\ref{rem:adjoints}, the ``singular'' creation operator defined as the adjoint with respect to the pairing corresponds to an extension of the ``regular'' one for Eq.~\eqref{eq:adaggera}.

    We shall need the following definitions (see, e.g.,~\cite[Section 6.1]{lampart2018particle}). Every single-particle vector $g\in\hilb$ is associated with a \textit{(non-normalized) coherent} vector $\varepsilon(g)\in\fock$ defined by
    \begin{equation}
	   \varepsilon(g):=\left\{\varepsilon(g)^{(n)}\right\}_{n\in\mathbb{N}},\qquad\varepsilon(g)^{(n)}=\frac{1}{\sqrt{n!}}g^{\otimes n}.
    \end{equation}
    In particular, if $g\in\hilb_{+s}$, then $\varepsilon(g)\in\fock_{+s}$ and, given $f\in\hilb_{-s}$, $\a{f}\varepsilon(g)=\braket{f,g}_{\hilb_{-s},\hilb_{+s}}\varepsilon(g)$. With this definition, given any subset $\mathcal{D}\subset\hilb$, the \textit{coherent domain over $\mathcal{D}$} is the subspace $\mathcal{E}(\mathcal{D})\subset\fock$ defined by
    \begin{equation}
	   \mathcal{E}(\mathcal{D})=\span\left\{\varepsilon(g):\;g\in\mathcal{D}\right\},
    \end{equation}
    and it can be proven that $\mathcal{E}(\mathcal{D})$ is dense in $\fock$ whenever $\mathcal{D}$ is dense in $\hilb$, and that $\varepsilon(g_n)\to\varepsilon(g)$ whenever $g_n\to g$.

    \subsection{Denseness of kernel and range}

    Statement (a) will follow from (a') in Theorem~\ref{thm:posilicano2}. As discussed above, in case $s\le1$, the operator $T = -ARA^*$ is bounded, and thus also $H_0$-bounded with relative bound zero. It remains to establish the denseness of $\ker A|_{\hfrak_2}$ and $\ran A|_{\hfrak_2}$, which we do in this subsection.
    
	\begin{lemma}\label{lem:kerrandense}
	    Let $f\in\hilb_{-s}\setminus\hilb$ for $s\in(0,2]$, and $\a{f}\in\mathcal{B}(\fock_{+s},\fock)$. Then 
	    \begin{itemize}
		    \item[(i)] $\ker\a{f}|_{\fock_{+2}}$ is dense in $\fock$;
		    \item[(ii)] $\ran\a{f}|_{\fock_{+2}}$ is dense in $\fock$.
	    \end{itemize}
    \end{lemma}
    
    \begin{proof} It suffices to prove the claims in the case $s=2$; the claims for $s<2$ are then true \textit{a fortiori}. To simplify the notation, in the following we will omit the subscript in the duality pairing $\braket{\cdot,\cdot}_{\hilb_{-2},\hilb_{+2}}$ between $\hilb_{-2}$ and $\hilb_{+2}$.
	
	(i) Since the functional $\hilb\supset\hilb_{+2}\ni\varphi\mapsto\braket{f,\varphi}$ is unbounded as a functional on $\hilb$, its kernel, i.e., the space $\mathcal{D}_0=\{g\in\hilb_{+2}:\,\braket{f,g}=0\}$, is dense in $\hilb$~\cite[Lemma~1.2.3]{albeverio2000singular}, whence its associated coherent domain $\mathcal{E}(\mathcal{D}_0)$ is dense in $\fock$. But, given any $g\in\mathcal{D}_0$, its corresponding coherent vector $\varepsilon(g)$ satisfies $\a{f}\varepsilon(g)=\braket{f,g}\,\varepsilon(g)=0$, whence the same holds for all elements of $\mathcal{E}(\mathcal{D}_0)$. This proves $\mathcal{E}(\mathcal{D}_0)\subset\ker\a{f}|_{\fock_{+2}}$. Since the former is dense in $\fock$, so is the latter.
	
	(ii) Since $f\neq0$, the set $\mathcal{D}=\left\{g\in\hilb_{+2}:\,\braket{f,g}\neq0\right\}$ is also dense in $\hilb$, thus $\mathcal{E}(\mathcal{D})$ is again dense in $\fock$. Given $g\in\mathcal{D}$, we then have $\a{f}\left(\braket{f,g}^{-1}\varepsilon(g)\right)=\varepsilon(g)$. This proves $\mathcal{E}(\mathcal{D})\subset\ran\a{f}|_{\fock_{+2}}$ and thus, again, the latter is dense in $\fock$.
    \end{proof}
    
    We remark that, while the denseness of the range of $\a{f}$ holds under the only condition $f\neq0$, the denseness of its kernel additionally requires $f\notin\hilb$.

	\begin{proposition}\label{prop:kerA}
        Let $s\leq2$, $f_1,\dots,f_N\in\hilb_{-s}\setminus\hilb$ a family of $\hilb$-independent form factors (see Definition~\ref{def:hilbindepentent}), and $B_1,\dots,B_N\in\mathbb{C}^{D \times D}$. Then, for $A$ defined in Eq.~\eqref{eq:AAdag}, $\ker A|_{\hfrak_{+2}}$ is dense in $\hfrak$.
    \end{proposition}	

    \begin{proof}
    Again, it suffices to prove the claim for $s=2$. Since $\{f_j\}_{j=1}^N$ are $\hilb$-independent, the space of single-particle states that are in the kernel of \textit{all} form  factors,
    \begin{equation}
        \mathcal{D}_0=\left\{g\in\hilb_{+2}\,\subset\hilb:\;\braket{f_j,g}_{\hilb_{-2},\hilb_{+2}}=0 \;\; \forall j=1,\dots,N\right\}
    \end{equation}
    is dense in the single-particle space $\hilb$, see~\cite[Lemma~3.1.1]{albeverio2000singular}. Then the coherent domain, $\mathcal{E}(\mathcal{D}_0)$, is also dense in $\fock$, and likewise $\mathfrak{h}\otimes\mathcal{E}(\mathcal{D}_0)$ is dense in $\hfrak$. Now, given any coherent vector $\varepsilon(g)\in\mathcal{E}(\mathcal{D}_0)$, one has $\a{f_j}\varepsilon(g)=0$, so for any $u\in\mathfrak{h}$,
    \begin{equation}
        A\,(u\otimes\varepsilon(g))=\sum_{j=1}^N(B_j^* u)\otimes(\a{f_j}\varepsilon(g))=0,
    \end{equation}
    whence, by linearity, $\mathfrak{h}\otimes\mathcal{E}(\mathcal{D}_0)\subset\ker A|_{\hfrak_{+2}}$. Since the former is dense in $\hfrak$, so is the latter.
    \end{proof}

    We remark that the denseness of $\ker A|_{\hfrak_{+2}}$ is irrespective of the choice of $B_1,\dots,B_N \in \mathbb{C}^{D \times D}$. No point of Assumption~\ref{as:Bj} enters here. However, the $\hilb$-independence of the form factors is crucial.

    \begin{proposition}\label{prop:ranA}
        Let $s\leq2$, $f_1,\dots,f_N\in\hilb_{-s}\setminus\hilb$ be a family of $\hilb$-independent form factors (see Definition~\ref{def:hilbindepentent}), and let $B_1,\dots,B_N\in\mathbb{C}^{D \times D}$ satisfy Assumption~\ref{as:Bj}. Then, for $A$ defined in Eq.~\eqref{eq:AAdag}, $\ran A|_{\hfrak_{+2}}$ is dense in $\hfrak$.
    \end{proposition}

    \begin{proof}
    Because of points (i) and (ii) in Assumption~\ref{as:Bj}, the matrices $B_1, \ldots, B_N$ share a common orthonormal eigenbasis: there exists a set of orthonormal vectors $\{u^{(a)}\}_{a=1}^D\subset\mathfrak{h}$ such that
    \begin{equation}
        B_ju^{(a)}=b_j^{(a)}u^{(a)},\qquad \forall j=1,\dots,N,\;\;a=1,\dots,D,
    \end{equation}
    with $b_j^{(a)}\in\mathbb{C}$. By normality, $\{u^{(a)}\}_{a=1}^D$ is also a common eigenbasis for their adjoints, with 
    \begin{equation}
        B_j^*u^{(a)}=\overline{b_j^{(a)}}u^{(a)},\qquad \forall j=1,\dots,N,\;\;a=1,\dots,D.
    \end{equation}
    We can thus conveniently decompose the Hilbert space as
    \begin{equation}
        \hfrak\simeq\bigoplus_{a=1}^D\hfrak^{(a)},\qquad
        \hfrak^{(a)}=\span\{u^{(a)}\}\otimes\fock\simeq\fock.
    \end{equation} 
    This decomposition is preserved by $A$, i.e., $A\hfrak^{(a)}\subset\hfrak^{(a)}$, since for any $\Psi\in\fock$,
    \begin{eqnarray}\label{eq:commuting}
        A\,(u^{(a)}\otimes\Psi)&=&\sum_{j=1}^N B_j^* u^{(a)}\otimes\a{f_j}\Psi\nonumber\\
        &=&\sum_{j=1}^N \overline{b_j^{(a)}}u^{(a)}\otimes\a{f_j}\Psi\nonumber\\
        &=&u^{(a)}\otimes\a{\sum_{j=1}^N b_j^{(a)}f_j}\Psi.
    \end{eqnarray}
    So we can also decompose
    \begin{equation}
        A=\bigoplus_{a=1}^N A^{(a)}, \qquad
        A^{(a)} = \a{\sum_{j=1}^N b_j^{(a)}f_j}.
    \end{equation}
    Obviously, $\ran A=\bigoplus_a\ran A^{(a)}$. So it suffices to prove that $\ran A^{(a)}$ is dense in $\hfrak^{(a)}$ for each $a$, in order to establish the proposition. Now, $A^{(a)}$ is just an annihilation operator on $\hfrak^{(a)}\simeq\fock$, so Lemma~\ref{lem:kerrandense} applies, provided that
    \begin{equation}
        \sum_{j=1}^N b_j^{(a)}f_j \in \hilb_{-2} \setminus \hilb.
    \end{equation}
    This requirement is fulfilled because of the $\hilb$-independence of $f_1, \ldots, f_N$, provided that there exists at least one $j=1,\dots,N$ such that $b_j^{(a)}\neq0$. The latter statement means that none of the vectors $u^{(a)}$ is such that $B_ju^{(a)}=0$ for all $j=1,\dots,N$ simultaneously, which is guaranteed by point (iii) of Assumption~\ref{as:Bj}, viz. $\bigcap_{j=1}^N\ker B_j=\{0\}$. Therefore Proposition~\ref{lem:kerrandense} indeed applies, rendering denseness of $\ran A^{(a)}|_{\fock_{+2}} \subset \ran A^{(a)}$ in $\hfrak^{(a)}$, which establishes the proposition.    
    \end{proof}

    \subsection{Finishing the proof}    
    
    \begin{proof}[Proof of Theorem~\ref{thm:main}]
    (a) We would like to apply Theorem~\ref{thm:posilicano2}. Since $f_j \in \hilb_{-s} \setminus \hilb$ for $s\le1$, we have $f_j \in \hilb_{-1} \setminus \hilb$. By Proposition~\ref{prop:af}, this implies $a(f_j) \in \mathcal{B}(\fock_1, \fock)$. In case of the spin--boson model~\eqref{eq:model} we are considering, it is then immediate to see that $A \in \mathcal{B}(\hfrak_{+1},\hfrak)$, and by duality, $A^* \in \mathcal{B}(\hfrak,\hfrak_{-1})$.  The requirement that $\ker A|_{\hfrak_2}$ and $\ran  A|_{\hfrak_2}$ be dense in $\hfrak$ have been proven in Propositions~\ref{prop:kerA} and~\ref{prop:ranA}, respectively.
    
    Now, recall that, for $s \in (0,1]$, the choice $T=-ARA^*$ is needed to make $H=H_{\textnormal{free}}+A^*+A$ agree with $H_T$ as defined before, cf.~Eqs.~\eqref{eq:at1}--\eqref{eq:at2}. This $T$ is a bounded operator in $\hfrak$, see Eq.~\eqref{eq:Tbound}, so it is in particular Kato-bounded with respect to $H_0$. Thus, Theorem~\ref{thm:posilicano2} applies, rendering statement (a).

    (b) We would like to apply Theorem~\ref{thm:posilicano3}. Conditions (i) and (ii) have already been verified. Proposition~\ref{prop:af} now provides us, for each form factor $f_j$, with a sequence of regularized form factors $\{f_{j, n}\}_{n \in \mathbb{N}} \subset \hilb$ such that $\Vert f_{j, n} - f_j \Vert_{\hilb_{-s}} \to 0$ as $n \to \infty$, and
    \begin{equation}\label{eq:aconvergence}
        \Vert a(f_{j, n}) - a(f_j) \Vert_{\fock_s \to \fock} \to 0, \qquad
        \Vert a^*(f_{j, n}) - a^*(f_j) \Vert_{\fock \to \fock_{-s}} \to 0.
    \end{equation}
    Correspondingly, we define the regularized operators as in Eq.~\eqref{eq:AnAndagger}
    \begin{equation}
        A_n := \sum_{j=1}^N B_j^* \otimes a(f_{j, n}), \qquad
        A_n^* := \sum_{j=1}^N B_j \otimes a^*(f_{j, n}),
    \end{equation}
    which are obviously closable as $f_{j, n} \in \hilb$. Furthermore, for $\Psi \in \hfrak_1$, we introduce the number operator $\mathcal{N}:=\mathrm{d}\Gamma(1)$, and we recall the following property (see e.g.~\cite[Proposition 2.3]{lonigro2022generalized}): for all $\Phi\in\fock_{+1}$,
    \begin{equation}
        \Vert\dOmega^{1/2} \Phi\Vert_{\fock}\geq m_0^{1/2}\Vert \mathcal{N}^{1/2}\Phi\Vert_{\fock}
    \end{equation}
    thus implying
    \begin{equation}
\begin{aligned}
     \Vert(\dOmega+1)^{1/2} \Phi\Vert^2_{\fock}&=\Vert\dOmega^{1/2}\Phi\Vert^2_{\fock}+\Vert\Phi\Vert^2_{\fock}
     \geq m_0\Vert \mathcal{N}^{1/2}\Phi\Vert^2_{\fock}+\Vert\Phi\Vert^2_\fock\\
     &\geq\min\{m_0,1\}\left(\Vert \mathcal{N}^{1/2}\Phi^2\Vert_\fock+\Vert\Phi\Vert^2_\fock\right)
     =\min\{m_0,1\}\Vert(\mathcal{N}+1)^{1/2}\Phi\Vert^2_\fock,
\end{aligned}
\end{equation}
whence, for every $\Psi\in\hfrak_{+1}$,
\begin{equation}\label{eq:in}
    \Vert(H_{\textnormal{free}}+1)^{1/2}\Psi\Vert_\hfrak\geq\Vert I_{\mathfrak{h}}\otimes(\dOmega+1)^{1/2}\Psi\Vert_\hfrak\geq\min\{m_0,1\}^{1/2}\Vert I_{\mathfrak{h}}\otimes(\mathcal{N}+1)^{1/2}\Psi\Vert_\hfrak.
\end{equation}
Therefore, denoting either $A$ or $A^\dag$ by $A^\sharp$, we have
\begin{equation}\label{eq:Anbound}    \begin{aligned}
        \Vert A_n^\sharp \Psi \Vert_{\hfrak}
        \le &\sum_{j=1}^N \Vert B_j \Vert_{\mathfrak{h}\rightarrow\mathfrak{h}} \Vert (I_{\mathfrak{h}} \otimes a^\sharp(f_{j, n})) \Psi \Vert_{\hfrak}
        \le \sum_{j=1}^N \Vert B_j \Vert_{\mathfrak{h}\rightarrow\mathfrak{h}} \Vert f_{j, n} \Vert_{\hilb} \Vert I_{\mathfrak{h}}\otimes(\mathcal{N}+1)^{\frac 12} \Psi \Vert_{\hfrak}\\
        \le &\bigg(\sum_{j=1}^N \Vert B_j \Vert_{\mathfrak{h}\rightarrow\mathfrak{h}} \Vert f_{j, n} \Vert_{\hilb} (\min\{m_0, 1\})^{-\frac 12}\bigg) \Vert (H_{\textnormal{free}}+1)^{\frac 12} \Psi \Vert_{\hfrak}
        =: c \Vert \Psi \Vert_{\hfrak_1},
    \end{aligned}
    \end{equation}
    %for some $c > 0$, 
    where we used Eq.~\eqref{eq:in}; so $A_n|_{\hfrak_1} \in \mathcal{B}(\hfrak_1, \hfrak)$ and $\dom (A_n) \supset \hfrak_1$, which establishes requirement (iii). Furthermore, Eq.~\eqref{eq:Anbound} implies that there is some $C>0$, such that for any $\varepsilon>0$:
    \begin{equation}
        \Vert A_n \Psi \Vert_{\hfrak}
        \le C \Vert (H_{\textnormal{free}}+1)^{\frac 12} \Psi \Vert_{\hfrak}
        \le C \big( \varepsilon \Vert (H_{\textnormal{free}}+1) \Psi \Vert_{\hfrak} + \varepsilon^{-1} \Vert \Psi \Vert_{\hfrak} \big).
    \end{equation}
        So $A_n^* + A_n$ is infinitesimally Kato-bounded with respect to $(H_{\textnormal{free}}+1)$ and thus also with respect to $H_{\textnormal{free}}$, which establishes requirement (iv). Concerning requirement (v), the convergence~\eqref{eq:aconvergence} implies
    \begin{equation}\label{eq:Aconvergence}
        \Vert A_n - A \Vert_{\hfrak_s \to \hfrak}
        \le \sum_{j=1}^N \Vert B_j \Vert_{\mathfrak{h}\rightarrow\mathfrak{h}} \Vert a(f_{j, n}) - a(f_j) \Vert_{\fock_s \to \fock} \to 0,
    \end{equation}
    as $n \to \infty$. So in particular, $\Vert A_n - A \Vert_{\hfrak_2 \to \hfrak} \to 0$, which is requirement (v).
    
    Finally, for (vi), we just choose $E_n = 0$. By the same mapping property argument as in Eq.~\eqref{eq:Tbound}, $A_n R A_n^*$ is bounded, so it is relatively bounded with respect to any symmetric operator with relative bound 0. Further, $T$ is bounded, so
    \begin{equation}
    \begin{aligned}
        &\Vert A_n R A_n^* + E_n + T \Vert_{\hfrak \to \hfrak}
        = \Vert A_n R A_n^* - A R A^* \Vert_{\hfrak \to \hfrak}\\
        \le &\underbrace{\Vert A_n - A \Vert_{\hfrak_s \to \hfrak}}_{\to 0}
            \underbrace{\Vert R A_n^* \Vert_{\hfrak \to \hfrak_s}}_{< \infty}
            + \underbrace{\Vert A R \Vert_{\hfrak_{-s} \to \hfrak}}_{< \infty}
            \underbrace{\Vert A_n^* - A^* \Vert_{\hfrak \to \hfrak_{-s}}}_{\to 0}
        \to 0,
    \end{aligned}
    \end{equation}
    where $\Vert A_n^* - A^* \Vert_{\hfrak \to \hfrak_{-s}} \to 0$ follows from Eq.~\eqref{eq:aconvergence} by the same arguments as in~\eqref{eq:Aconvergence}. This also establishes (vi). 
    
    Consequently, Theorem~\ref{thm:posilicano3} applies, and we obtain the desired convergence~\eqref{eq:normresconvergence}.
    \end{proof}
    
	\section{Concluding remarks}\label{sec:conclusions}

    In this work, we have discussed the ultraviolet renormalization problem for generalized spin--boson models, exploiting the abstract framework provided by Posilicano in~\cite{posilicano2020self}. By this technique, under suitable requirements, we were able to provide an explicit expression for the self-adjointness domain of the Hamiltonian. The requirements include the choice of form factors $f_1,\dots,f_N\in\hilb_{-1}$ (Case~1, see Eq.~\eqref{eq:hilbscale}), so they exhibit ``mild'' ultraviolet divergences, as well as Assumption~\ref{as:Bj} on the operators $B_1,\dots,B_N$ modulating the interaction between the quantum system and the boson field. Furthermore, it was shown that such models can be obtained as the norm resolvent limit of their regularized versions (e.g., those in which an ultraviolet cutoff is imposed). This generalizes the results presented in~\cite[Section 4]{lonigro2022generalized}, where the existence of a self-adjoint realization of generalized spin--boson models was proven through perturbative methods. 
    
    The research initiated in this paper offers, among others, two clear routes of continuation. First, the conditions on the matrices $B_1,\dots,B_N$, while reasonably general and including many cases of physical interest, are likely prone to be relaxed. This is strongly suggested by the results in~\cite[Sections 5--6]{lonigro2022generalized} and~\cite{lonigro2023self}, where an explicit expression for the self-adjointness domain fully compatible with the one provided here (cf.~\cite[Remark 3.8]{lonigro2023self}) was obtained for a class of models corresponding to a choice of $B_1,\dots,B_N$ which \textit{violates} Assumption~\ref{as:Bj}. An extension of our results to the case of an infinite-dimensional quantum system interacting with the field is also foreseeable along the lines depicted in Remark~\ref{rem:infdim}.

    Second, the abstract framework developed by Posilicano is naturally suited to treat the case of ``strong'' ultraviolet divergences belonging to Case~2. As already discussed (cf.~Remark~\ref{rem:selfenergy}), additional care is required in this case in order to properly choose the operator $T$ entering the statements of Theorem~\ref{thm:posilicano1}--\ref{thm:posilicano3}: while for mild divergences the ``minimal'' choice $T=-ARA^*$ is available (and bounded), novel mathematical intricacies appear when renormalization enters the game. Future research will be devoted to this topic.\bigskip

\small
\noindent\textit{Acknowledgments.}
    This research was supported by the European Union (ERC FermiMath, grant agreement nr. 101040991 of Niels Benedikter). Views and opinions expressed are however those of the author(s) only and do not necessarily reflect those of the European Union or the European Research Council Executive Agency. Neither the European Union nor the granting authority can be held responsible for them. Further, this research was supported by ``Istituto Nazionale di Fisica Nucleare'' (INFN) through the project ``QUANTUM'', by the Italian National Group of Mathematical Physics (GNFM--INdAM), and by European Union – NextGenerationEU (CN00000013 – “National Center for HPC, Big Data and Quantum Computing”). D.L. thanks the Department of Mathematics at ``Università degli Studi di Milano'' for its hospitality. \normalsize
   
\AtNextBibliography{\small}	\DeclareFieldFormat{pages}{#1}\sloppy 
	\printbibliography
	
\end{document}